%% file: main.tex
\documentclass{elsarticle}

\usepackage{amsmath,amssymb,amsfonts}
\usepackage{algorithmic}
\usepackage{graphicx}
\usepackage{textcomp}
\usepackage{xcolor}
\usepackage[utf8]{inputenc}
\usepackage{mathpartir,color}
\usepackage{stmaryrd}
\usepackage{amsthm}
\usepackage{graphicx}

\newcommand{\symb}[1]{\textit{#1}} 
\newcommand{\noop}{\symb{Noop}}
\newcommand{\Push}{\symb{Push}}
\newcommand{\Pull}{\symb{Pull}}
\newcommand{\while}{\symb{While}}
\newcommand{\cond}{\symb{cond}}
\DeclareMathOperator{\vars}{vars}
\newcommand{\isvalid}{\symb{isValid}}
\newcommand{\isremvalid}{\symb{remIsValid}}
\newcommand{\rem}[1]{\symb{rem}(#1)}
\newcommand{\IF}[3]{\symb{if}\,(#1)~#2~\symb{else}~#3 }
\newcommand{\feval}[2]{\llbracket#1\rrbracket_{#2}}
\newcommand{\True}{{\tt True}}						
\newcommand{\False}{{\tt False}}				
\newcommand{\transl}[1]{\llbracket#1\rrbracket}
\newtheorem{definition}{Definition}
\newtheorem{Property}{Property}
\newtheorem{Theorem}{Theorem}
\newcommand{\abs}[1]{#1^\#}
\newcommand{\AM}{\mathcal{M}}
\newcommand{\Prog}{\mathcal{P}}
\usepackage{etoolbox}
\DeclareMathOperator{\range}{index}
\DeclareMathOperator{\dom}{dom}
\newcommand{\Overlap}[1]{O(#1)}

\newenvironment{example}{\smallskip\par\noindent\emph{Example:}}{\medskip}

\usepackage{lineno}

\journal{Journal of Logical and Algebraic Methods in Programming}

\begin{document}

\begin{frontmatter}


\title{Leveraging Access Mode 
Declarations in a Model for Memory Consistency in Heterogeneous Systems }

\author[lip]{Ludovic Henrio\fnref{cor2}\corref{cor1}}
\ead{ludovic.henrio@cnrs.fr}
\author[liu]{Christoph Kessler}
\ead{christoph.kessler@liu.se}
\author[liu]{Lu Li}
\ead{lu.li@liu.se}

\fntext[cor2]{This work has been mostly founded and realised while Ludovic Henrio was affiliated to: Universit\'e~C\^ote~d'Azur, CNRS, I3S, France.}
\cortext[cor1]{Corresponding author}
\address[lip]{Univ Lyon, CNRS, ENS de Lyon, Inria, Universit\'e Claude Bernard Lyon 1, LIP, Lyon, France.}
\address[liu]{University of Linköping, Sweden}

\begin{abstract}
On a system that exposes disjoint memory spaces to the software, a program has to address memory consistency 
issues and perform data transfers so that it always accesses valid  
data. Several approaches exist to ensure the consistency of the memory accessed. Here we 
are interested in the verification of a declarative approach where each component of 
a computation is annotated with an access mode declaring which part of the memory is 
read or written by the component. The programming framework uses the component
annotations to guarantee the validity of the memory accesses. This is the mechanism used in VectorPU, a C++ library for programming CPU-GPU heterogeneous  systems. 
This article proves the correctness of the software cache-coherence mechanism used in VectorPU. 
Beyond the scope of VectorPU, this article provides a simple and effective formalisation of memory consistency mechanisms based on the explicit declaration of the effect of each component on each memory space. 
The formalism we propose also takes into account arrays for which a single validity status is stored for the whole array; additional mechanisms for dealing with overlapping arrays are also studied. 

\end{abstract}

\begin{keyword}
Memory consistency \sep CPU-GPU heterogeneous systems \sep  
data transfer \sep  software caching \sep  cache coherence 


\end{keyword}

\end{frontmatter}

\section{Introduction}

\input{vectorpu-intro.tex}

\section{A Formalization for Reasoning on  Consistency in VectorPU}\label{sec:Formal}

In this section we provide a minimal calculus to reason on the memory operations that can 
exist in a framework that deals with memory consistency like VectorPU. We first define a 
set of effects that operations can have on the consistency of the memory. Then we define a 
small calculus expressing different memory accesses and their composition into complex 
procedures. Finally, we express VectorPU operations as higher-level statements that can 
be translated  into the core calculus, and show that if all memory 
accesses are annotated correctly through VectorPU annotations the program cannot try to 
access an invalid data and the memory spaces are put in coherence when needed. We also 
show that VectorPU tracks the validity status of the memory adequately. In this 
section we abstract away the values stored in memory and we 
do not deal with any form of aliasing. A more precise analysis of aliasing is 
out of the scope of this paper, it could be for example inspired 
from~\cite{Nielson1999} or from our extension to overlapping array (Section~\ref{sec:overlap-array}).
We place ourself in a simplified setting where each variable is hosted in exactly two 
memory locations, e.g.\ a CPU (main) memory and a GPU memory location; our results could be extended to multiple memory 
locations without any major difficulty.

\subsection{An effect system for consistency between memory locations}
We start from a simple effect system, it expresses the effect of writing or reading a 
memory 
location on the consistency status of the memory. Each location is either in  
\textit{valid} state when it holds a usable data 
or \textit{invalid} state when the value at the location is not valid anymore.
We express five  operations: reading, writing, 
\Push\ for uploading the local memory location into the other one, and \Pull\ for the 
contrary. \noop\  does nothing.
 \[E::= \Push ~|~ \Pull ~|~ r~|~ w ~|~ 
\noop\]

These operations involve a single memory 
location. We 
express 
below the semantics of each of the operations on the consistency status of the concerned 
memory location. 
The \textit{memory status of a variable}
is a pair of the status of its locations, 
where each status is 
either $V$ for valid or $I$ for invalid. The first element is the status of the local 
memory, and the second one is the status of the remote memory. For example, for a 
program running on a CPU while the remote memory is a GPU, a status $(V,I)$ means that 
the memory is valid and can be read on the CPU, but is invalid on the GPU and should be 
transferred before being usable there.
Each operation has a signature in the sense that it may require a certain memory status 
and 
will produce another memory status. The signature of each operation is expressed below and is called its \emph{effect}.
We use
variables --$X$, $Y$-- that are considered universally quantified in each rule. 
They can 
be instantiated with either $V$ or $I$.
\begin{mathpar}
\Push: (V,X)\mapsto (V,V)

\Pull: (X,V)\mapsto (V,V)

r: (V,X)\mapsto (V,X)


w: (X,Y)\mapsto (V,I)

\noop: (X,Y)\mapsto (X,Y)
    \end{mathpar}

These signatures are effects  expressing that
$r$ is a reading operation requiring validity of data and ensuring not to modify it, 
the distant status is unchanged; $w$ 
is a  writing operation that modifies data locally but do not require validity, it 
invalidates the remote memory.  \Push\ uploads the local memory and thus makes valid the 
distant memory; 
it requires that the data is locally valid, and \Pull\ is the symmetrical operation.
Applying these signature consists in trying to \emph{unify} the current memory status with the effect of the variable, potentially instantiating $X$ and $Y$ appropriately. No unification is possible if the status and the effect cannot be made identical by instantiating variables.

\begin{example}
An operation $r$ can be applied on a validity status $(V,V)$, leading to the validity status $(V,V)$ because $(V,V)$ can be unified with $(V,X)$ by instantiating $X$ with $V$. However $r$ can not be applied on a validity status $(I,V)$ because $(V,X)$ cannot be unified with $(I,V)$: $I$ and $V$ are different.
\end{example}

An additional operation could be defined: 
 an 
$rw$ operation would represent a read and/or write access, it would both require data 
validity and invalidate 
remote status: $(V,X)\mapsto (V,I)$. This operation is however not needed here but we will have a similar one at the annotation level, see below.

\subsection{A language for modelling consistency and effects}\label{sec-core}
We now create a core calculus to  reason on programs involving sequences of 
effects on different memory locations. $x,y$ range over variables and we introduce  statements manipulating 
variables. We use sequence and simple loops and conditionals. 
Operations with effects 
now apply to a variable: $E$ $x$ denotes some operation $E$ on variable $x$;   $\rem{E~x}$  is a remote operation 
 on the remote memory. We also have a $\noop$ operation that has no effect and can be considered as a neutral element for the sequence.
Statements $S$ are defined as:
{\small \[S::=E~x~|~\rem{E~x}~|~S;S'~|~\while(\cond) S ~|~ \IF{\cond}{S}{S'}~|~\noop\]}
\vspace{-3.5ex}
\begin{example}
A GPU procedure 
writing $x$ and reading $y$ would correspond to the pseudo-code: $\rem{w~x};\rem{r~y}$. \end{example}

We are  interested  in conditionals dealing 
with the validity status of the variables. Other conditionals  are expressed as a generic 
binary operator $\oplus$ but  operators with different arities could be added as 
well:
\[\cond::=\isvalid~x~|~\isremvalid~x~|~x\oplus y\]
\noindent where $\isvalid~x$, resp. $\isremvalid~x$,  denote checks of the validity status flag of the local, resp. remote, location of $x$.

We now define a small step operational semantics for our core calculus.
It relies on the validity status of variables, recorded in a store $\sigma$ mapping 
variable names to validity pairs. Semantics is written as a transition relation between 
pairs consisting of a statement and a store: $(S,\sigma)$. The sequencing operator $;$ is associative with \noop\ as a neutral element. 
Consequently each non-empty sequence of instruction can be rewritten as $S;S'$ where $S$ 
is neither a 
sequence nor \noop. $\sigma[x\mapsto (X,Y)]$ is the update operation on maps. 

\begin{figure*}[t]
\begin{small}
\begin{mathpar}
\inferrule[valid]
{\sigma(x)=(V,X)}
{\feval{\isvalid~x}{\sigma}=\True}

\inferrule[invalid]
{\sigma(x)=(I,X)}
{\feval{\isvalid~x}{\sigma}=\False}

\inferrule[rem-valid]
{\sigma(x)=(X,V)}
{\feval{\isremvalid~x}{\sigma}=\True}

\inferrule[rem-invalid]
{\sigma(x)=(X,I)}
{\feval{\isremvalid~x}{\sigma}=\False}

\inferrule[Effect]
{\sigma(x)=(X,Y) \\ E:  (X,Y)\mapsto (Z,T)  }
{(E~x;S,\sigma)\to (S,\sigma[x\mapsto(Z,T)])}

\inferrule[Remote Effect]
{\sigma(x)=(X,Y) \\ E:  (Y,X)\mapsto (Z,T) }
{(\rem{E~x};S,\sigma)\to (S,\sigma[x\mapsto(T,Z)])}

\inferrule[While-True]
{\feval{\cond}{\sigma} }
{(\while(\!\cond) S;\!S',\sigma)\!\to\! (S;\!\while(\!\cond) S;S',\sigma)}
\qquad
\inferrule[While-False]
{\neg\feval{\cond}{\sigma} }
{(\while(\!\cond) S;\!S',\sigma)\!\to\! (S',\sigma)}

\inferrule[IF-True]
{\feval{\cond}{\sigma} }
{((\IF{\!\cond} S S');\!S'',\sigma)\!\to\! (S;\!S'',\sigma)}
\quad~
\inferrule[IF-False]
{\neg\feval{\cond}{\sigma} }
{((\IF{\!\cond} S S');\!S'',\sigma)\!\to\! (S';\!S'',\sigma)}
    \end{mathpar}
\end{small}
\caption{Operational semantics of validity status.}\label{fig:Opsem}
\end{figure*}

The semantics is presented in Figure~\ref{fig:Opsem}. Like in the previous section, we 
use validity variables $X$, $Y$, $Z$, $T$ 
that are universally quantified in each rule.
 The first 
 four rules present the 
evaluation of conditional statements, we assume additional rules exist for evaluating 
$\oplus$\footnote{We  suppose  that evaluation of $\oplus$ always succeeds, and in particular  variables accessed by the operation are specified as a $r$ operation preceding the condition.}. The next rule applies an effect of the operation $E$ on a variable $x$ updating the validity store, 
and the \textsc{Remote Effect} rule applies an operation occurring on the distant memory, it 
applies the symmetric of the effect of the operation to the variable: \emph{validity values are switched compared to the non-remote effect}. Note that \Push\ is the symmetric of 
\Pull\ and we could have removed one of those two operations without loss of generality: $\Pull$ is the same as $\rem \Push$. 
The last rules are standard ones for \symb{if} and \symb{while} statements.

\noindent\emph{Initial state:} To evaluate a sequence of statements $S$ using the 
variables 
$\vars(S)$, we create an initial configuration with 
a store where data is hosted on the CPU and all variables are  mapped to $(V,I)$: 
$\sigma_0=(x\mapsto 
(V,I))^{x\in \vars(S)}$.


A configuration is \emph{reachable} if it is possible to obtain this configuration 
starting from the initial configuration and applying any number of reductions: 
$(S,\sigma)$ is reachable if $(S,\sigma_0)\to^*(S',\sigma)$ where $\to^*$ is the 
reflexive 
transitive closure of $\to$. We write $(S,\sigma)\not\to $ and say that the configuration 
is \emph{stuck} if no reduction rule can be 
applied on $(S,\sigma)$.
\begin{Property}[Progress]\label{prop:stuck}
 A configuration is stuck if the validity status of the accessed variable is 
incompatible with the effect to be applied:\\[-1ex]
\[\begin{array}{@{}l@{}}
(S,\sigma)\not\to\ \iff~
 \begin{array}[t]{@{}l@{}}
S=E~x;S' \land \sigma(x)=(X,Y) 
							\land E:  (X',Y')\mapsto (Z,T)\ 
						\\\qquad 	\land~\text{there is no unification  between } 
							(X,Y)   
							\text{ and } (X',Y')\\
\lor S\!=\!\rem{E~x};S' \land \sigma(x)\!=\!(X,Y) 
							\land E:  (X',Y')\!\mapsto\! (Z,T)\ 
							\\\qquad \land~\text{there is no unification  between } 
							(X,Y)   
							\text{ and } (Y',X')	\\
\lor S\!= \noop	
\end{array}
\end{array}\]
Note that this supposes that $\oplus$ always succeeds.
\end{Property}
\begin{proof}
Recall each non-empty sequence of instruction, here $S$,  can be rewritten as $S';S''$ where $S$ 
is neither a sequence nor \noop.
If the sequence is empty, $S=\noop$, and the execution is finished, it corresponds to the last case of the rule.

By case analysis on the first statement of $S'$, there is always 
one rule applicable provided the premises of the rule can be evaluated.

In  case the statement is an if or a while, it corresponds to the last four rules this requires the evaluation of \cond. If $\oplus$ 
always succeeds then \cond\ can always be 
evaluated to either \True\ or \False\ and consequently one of the rule can always be applied.

 The only cases remaining for $S$ are $E~x$ and $\rem {E~x}$. The applicable rules are \textsc{Effect} and \textsc{Remote Effect}. These rules can always be applied except if there is no unification possible between the 
effect of an operation and the current validity status of the affected variable, i.e. there is no instantiation of $X$ and $Y$ such that both $\sigma(x)=(X,Y)$ and $E: (X,Y)\to(Z,T)$ in the case of \textsc{Effect}. This  corresponds to the two first
cases expressed in the theorem (one for \textsc{Effect} and one for \textsc{Remote Effect}).
\end{proof}

\begin{Property}[Safety]\label{prop:safe}
A state is said to be \emph{unsafe} if at least one variable is mapped to 
$(I,I)$.
It is impossible to reach an unsafe state from the initial state.
\end{Property}
\begin{proof}
The principle is that unsafe states are avoided  because of the effects of operations. Indeed 
only the two effect rules (\textsc{Effect} and \textsc{Remote Effect}) modify the store and no effect can reach $(I,I)$, except 
$\noop\ x$ starting from $\sigma(x)=(I,I)$. This is sufficient to conclude, by recursion, as the initial state is not $(I,I)$.
\end{proof}

\begin{example} The sequence
   $w~ x;\rem{r~x}$ can never be fully evaluated and will lead to a stuck configuration. Indeed, $(w x;\rem{r~x},(x\mapsto 
    (V,I)))\to (\rem{r~x},(x\mapsto (V,I)))$, but  $\rem{r~x}$ requires that $x$ is 
    mapped to $(X,V)$ for some $X$ which is not the case.

However if we add a \Push\ operation to ensure the validity of the accessed memory the 
program $w~x;\Push~x;\rem{r~x}$ can be reduced as follows:\\
$(w x;\rem{r~x},(x\mapsto (V,I)))\\~\qquad\to (\Push~x;\rem{r~x},(x\mapsto 
(V,I)))\\~\qquad\to 
(\rem{r~x},(x\mapsto (V,V)))\to (\noop,(x\mapsto (V,V)))$
\end{example}

\subsection{Declaring access modes and adding an abstraction layer}
The calculus defined above only considers simple memory locations and directly manipulates 
them.
But VectorPU and  similar libraries manipulate structures 
representing the memory. For example, VectorPU vectors act as an
 abstract representation of a set of memory 
locations. In this section, we add a declaration and abstraction layer to the calculus to 
represent the access mode declarations that will trigger data transfers according to the 
consistency mechanism. 
This abstraction layer is also a necessary first step to the modelling of array 
structures that we will present in Section~\ref{sec-arrays}. Indeed, in array structures, 
the 
validity status of the array is abstracted away by a single validity status pair. Then 
a dynamic abstraction of the consistency status of the memory can be used.
Abstract variables are not part of the applicative code but can be used in the access mode declarations. 

Consider for example an array $x$. As it is not desirable to store the information of the validity status of all the elements of the array, we will create for this array and abstract variable $\abs x$ that will represent all the elements of the array -- $x[i]$ -- from the validity status point of view. If the representation was precise, a value $(V,I)$ for $\abs x$ would mean that all the elements of the array are valid locally and invalid remotely. In practice we need to authorize some approximation of the validity status, at least because most operations only act on some of the elements of the array. Consequently, some  information is lost in the abstract representation: if $\sigma(\abs x)=(V,I)$ then each element of the array $x$ must be locally valid, the remote elements may be valid or invalid. Like with non-abstract variables, having $\sigma(\abs x)=(I,I)$ is  not desirable as the coherence protocol would be unable to know where are the valid elements of the array -- even if for all $i$, $x[i]=(V,I)$ or $x[i]=(I,V)$. In the case of arrays and  usual data-structures, the mapping between abstract and concrete elements is quite trivial: one abstract variable represents a whole data structure. Other mappings (one variables for several structures or splitting a data structure) could be defined but their  definition might be too complex to be usable in practice.

The abstraction and declaration layer relies on two principles:
\begin{itemize}
\item Each variable $x$ has an abstract variable $\abs x$ that represents it. In this 
section there is 
a single variable for each representative, but when we deal with arrays we will 
have a single representative for the whole array.
\item It is safe to ``forget'' that one memory space holds a valid copy of the data if 
the other memory space has a valid one. In other words, $(V,I)$ (resp. $(I,V)$) is a safe 
abstraction of $(V,V)$ and we have $(V,I)\!\leq\! (V,V)$ (resp. $(I,V)\!\leq\! (V,V)$). Also for any $X$ and $Y$ we have $(X,Y)\!\leq\! (X,Y)$.
\end{itemize}

\paragraph{Syntax}
We now define access mode declarations:\\[-3.3ex]
\begin{align*}
\AM&::=R\ \abs x \,|\, W\ \abs x \,|\, RW\ \abs x \,|\, 
\rem{R\ \abs x} \,|\,\rem{W\ \abs x} \,|\,\rem{RW\ \abs x} \,|\, \\
&~ \AM \land \AM' \quad \text{(where variables in $\AM$ and $\AM'$ are disjoint)}
\end{align*}

These access modes declare the kind of access (read $R$, write $W$, or read and/or write 
$RW$) that 
can be performed on the variable $x$ represented by $\abs x$. In a set 
of access mode declarations the same variable cannot appear twice\footnote{This restriction simplifies the formal definition and the reasoning. Extending the results to the same variable appearing twice with the same access mode is trivial (if the same parameter is passed twice to the same function). Having the same memory location declared twice with different access modes is not safe in the general case but Section~\ref{sec:overlap-array} will study more precisely the case of overlapping arrays.}. There exist declared 
access modes for  local accesses and for  the 
remote memory space.

A program is a sequence of calls to functions or components (i.e., statements accessing 
only real variables) 
each protected by an access 
mode declaration (on abstract variables representing the real variables):
\[\Prog::=\AM_1\{S_1\};\AM_2\{S_2\};\ldots\]
We write  $S\in S'$ if $S$ is one statement inside $S'$ (i.e. $S$ is a sub-term of 
$S'$).

We  define below the semantics of these programs  and specify well-declared program by 
comparing the statements they contain with 
the declared access modes. The semantics relies on the translation of 
the access mode declarations into consistency mechanisms 
with checks and data transfers 
triggered 
before each function 
execution.

\paragraph{Extension of statements to abstract variables}
When evaluating a program, the store contains both real and abstract variables, and the 
existing 
statements have the same effect on the abstract variables as on the real ones. However 
one should notice that 
even if the effect is the same, the meaning of a statement acting on a real variable 
or on its representative is different: in our calculus, the effect on a variable is an 
abstraction of the real effect that involves side effects and data transfers. On the 
contrary, only the validity status of abstract variables is stored by the library: the 
effect triggered by an operation on an abstract variable is exactly what happens when 
VectorPU updates the validity status of its internal structures.

For example, a \Pull\ operation on a real variable consists in transferring data from a 
remote memory space 
to the local one. We abstracted it  by changing the local validity status. A \Pull\  
operation on an abstract variable only changes the validity status, no data transfer has 
to be done because abstract variables only need to be stored in one memory space. 
The validity status is stored in the CPU address space in VectorPU. Comparing the validity 
status of real memory and their representative  allow us to reason 
formally on the correctness of the validity tracking performed by VectorPU.

\begin{figure*}[tb]
\begin{mathpar}
\transl{R~\abs x}=(\IF{\isvalid~\abs x}{\noop}{(\Pull~x;\Pull~\abs x)})

\transl{\rem {R~\abs x}}=(\IF{\isremvalid~\abs x}{\noop}{(\Push~x;\Push~\abs x)})

\transl{RW~\abs x}=(\IF{\isvalid~\abs x}{\noop}{(\Pull~x;\Pull~\abs x)});w~\abs x

\transl{\rem{RW\,\abs x}}\!=\!(\IF{\isremvalid~\abs x}{\noop}{({\!\Push~x};\Push~\abs x)});\rem{w~\abs 
x}

\transl{W~\abs x}= w~\abs x

\transl{\rem{W~\abs x}}= \rem{w~\abs x}

\transl{\AM_1\{S_1\};\AM_2\{S_2\};\ldots} = \transl{\AM_1};S_1;\transl{\AM_2};S_2;\ldots
\end{mathpar}
\caption{Semantics of access modes and programs}\label{sem-AM}
\end{figure*}
As no data is accessed by the effects on abstract variables, they cannot create stuck configuration. Consequently, $r~\abs x$ has no effect as it does not 
change the validity 
status of variables. The statement that should get stuck in case of a read access is the 
read of \emph{the real variable that cannot access a valid data}. 

\paragraph{Semantics}

Figure~\ref{sem-AM} defines the semantics of programs with access modes as a translation 
into the core calculus of Section~\ref{sec-core}. 
This translation
ensures that the validity status is correct and records the effect of the function on the 
abstract variable before running the function call that may 
read and write data (on the real variables).  Similarly to the VectorPU library, the 
protected accesses can be 
considered as macros and the programs can be translated into the 
core syntax.

This encoding corresponds  to the macros as they are implemented in VectorPU.  
It is indeed easy 
to check that VectorPU tracks the effects in 
the same way as our effect system  does in the translation rules. These translation 
rules  perform \Push\ or 
\Pull\ operations in order to ensure that the memory is in a correct validity status for 
the read or write operation to be performed.
The validity conditions are checked on the abstract value, which corresponds to the fact that VectorPU only check the status of the coherency flag stored with the vector structure; push/pull operations are performed twice: once for representing the data-transfer, and once for representing the validity status update. Finally, status is updated when writing operations are declared.
When evaluating a program we create a store where the validity status of real and 
abstract variables are $(V,I)$, corresponding to the fact that data is 
initially placed in one memory location; typically, in VectorPU, in the CPU memory space.
%
%

\subsection{Well-declared Programs and their Properties}
We now define formally what it means for an access mode declaration to be correct, i.e. 
to adequately specify the effect of a function. The principle is that each operation on a 
memory location must be declared on its representative. It is however possible to declare 
more read or $RW$ accesses that what is done in practice, and one can declare a read 
and/or write 
access if only read or write is performed. Additionally, the annotation $W$ denotes an 
\emph{obligation} to write which allows the consistency mechanism to avoid any validity 
check and any transfer before running the function that will overwrite the data. 
To represent this concept, we need a first definition that states that an operation will be performed in all execution paths of a (bigger) statement. This 
definition 
formalises a classical static analysis concept that states that all branches of 
conditionals necessarily execute a given statement. It considers executions that run to 
completion and states that a given statement is necessarily evaluated in this execution.

\begin{definition}[Occurs in all execution paths]
We state that a statement\emph{ $S$ occurs in all execution paths of $S_0$} if, for any 
correct 
initial store $\sigma_0$, for all full reductions 
$(S_0,\sigma_0)\to(S_1,\sigma_1)\to\ldots\to(\noop,\sigma_n)$, there is an intermediate 
state $(S_i,\sigma_i)$ such that $S_i=S;S''$ for some $S''$.
\end{definition}
Notice that an operation $S$ can only appear in some of the execution paths of $S'$ if  $S\in 
S'$: if $S$ is an operation, i.e. a single statement, then $(S_0,\sigma_0)\to^* (S;S',\sigma)$ then $S\in S_0$.

\begin{definition}[Well-declared program]\label{def-WD}
A program $\Prog$ is \emph{well-declared} if for all $\AM\{S\}$ in $\Prog$ we have:
\begin{itemize}
\item $\Push\ x\not\in S$ and $\Pull\ x\not \in S$ (for any $x$),
\item $w\ x\in S \implies (W\ \abs x \in \AM \lor RW\ \abs x \in \AM)$,
\item $r\ x\in S \implies (R\ \abs x \in \AM \lor RW\ \abs x \in \AM)$,
\item $W\ \abs x\!\in\! \AM \!\implies\! w\ x$ occurs in all execution paths of $S$,
\item Plus the same rules for remote operations.
\end{itemize}
\end{definition}
Note that a well-declared program does not perform synchronisation 
operations (\Push\ or \Pull) manually, these operations are only  performed when 
evaluating the 
access mode declarations. Also each variable accessed by a well-declared function has an 
abstract representative in the corresponding declaration block.

A direct consequence of the definition above is that a well-declared program cannot 
access, in the same function, the same
variable in both address spaces. This is in accordance with VectorPU where each function 
is entirely executed either on a CPU or on a GPU, the formalisation is a bit more 
 generic on this aspect. This is expressed by the following property.
\begin{Property}[Localised access]\label{prop-localised}
Consider a well-declared program containing $\AM\{S\}$, for any $x$, we cannot have $\rem {E\ 
x} \in S$ and $E'\ x \in S$.
\end{Property}
\begin{proof}
This is a consequence of the uniqueness of abstract variables in access mode declarations. Indeed, if $E'\ x\in S$ and the program is well-declared, then there must be $R\ \abs x$ or $W\ \abs x$ or $RW\ \abs x$ in $\AM$. Similarly, as $\rem {E\ x} \in S$, if the program is well-declared $\rem {R\ \abs x}$ or $\rem {W\ \abs x}$ or $\rem {RW\ \abs x}$ in $\AM$, it is impossible to have two different entries for the same variable and thus we cannot have $\rem {E\ 
x} \in S$ and $E'\ x \in S$.
\end{proof}

\smallskip

We now state and  prove the two major properties ensured by our formalisation.
The first property ensures that the abstraction is correct relatively to the 
execution. This corresponds to the fact that VectorPU tracks adequately the validity 
status of the 
memory. This is expressed as a theorem that is similar to subject-reduction in type 
systems, it states that if the status of the abstract variables represent correctly the 
validity status of the real variables, then the 
abstraction is also correct after the execution of a  well-declared function.  We first define the correctness of the representation of the validity status.
\begin{definition}[Correct abstraction of the memory state]\label{CorrectAbstraction}
We have a \emph{correct abstraction of the memory state} if for each real memory 
location, the abstract representative of this location has a validity status that is an 
approximation, in the sense of $\leq$, of the validity status of the real memory. More formally, $\sigma$ stores a correct abstraction of the memory state if (recall $\abs x$ stores the validity information for $x$):
\[ \forall x\in \dom(\sigma).\, \sigma(\abs x)\leq\sigma(x) \]
\end{definition}
 The 
theorem below states that the execution of a well-declared function maintains the 
correctness of the memory state 
abstraction. 

\begin{Theorem}[Subject reduction]\label{thm-SR}
Suppose $\AM\{S\}$ is well-declared, we have:
\[\begin{array}{@{}l@{}}
\forall x\in \dom(\sigma).\, \sigma(\abs x)\leq\sigma(x) 
\land
 (\transl{\AM\{S\}},\sigma) \to^* (\noop,\sigma')
\\~~~~~
\implies \forall x\in \dom(\sigma').\, \sigma'(\abs x)\leq\sigma'(x) 
\end{array}\]

This property is extended by a trivial induction to the execution of a well-declared 
program $\Prog$ in an initial store $\sigma_0=(x\mapsto 
(V,I))^{x\in \vars(\Prog)}$.
\end{Theorem}
\begin{proof}
Notice that $\dom(\sigma')=\dom(\sigma)$, and  if $\sigma(x)=(V,I)$ or $\sigma(x)=(I,V)$ 
then $\sigma(x)=\sigma(\abs x)$, else $\sigma(x)=(V,V)$.
We reason on the read and write access that occur in the considered reduction. Each 
variable $x$ is either read or written or not accessed (or read and written). For each 
case we compare the 
status of abstract and local variable, and in particular we consider the status of the 
reduction after executing the synchronisation code $\transl{\AM\{S\}}$ and call 
$\sigma_s$ the corresponding store (note that $\sigma_s(\abs x)=\sigma'(\abs x)$). We 
detail operations on the local 
address space, cases for remote operations are similar:

\noindent$\bullet$ If  $x$ is written, we have:
$(\transl{\AM\{S\}},\sigma)\to^* (w~x;S',\sigma'')  \to^* (\noop,\sigma')$. Whatever the 
initial value of $\sigma(x)$, we have $\sigma'(x)=(V,I)$. Two cases are possible:

\noindent$(1)$ $W~\abs x \in \AM$ then the value cannot be read and we have 
$\sigma_s(\abs x)=(V,I)$.  $ \sigma'(\abs x)=\sigma'(x)$.

\noindent$(2)$ $RW~\abs x \in \AM$ then a data-transfer (\Pull) may occur. Knowing that 
$\sigma(\abs x)\leq\sigma(x)$, by a  
case 
analysis on $\sigma(x)$ and $\sigma(\abs x)$ we have: $\sigma_s(\abs x)=(V,I)$ and 
$\sigma_s(x)=(V,I)$ or $(V,V)$. Whether $x$ is  read or not we have $ \sigma'(\abs 
x)=\sigma'(x)$.

\noindent$\bullet$ If  $x$ is read but not written, its validity status is 
not changed. 

\noindent$(1)$ $R~\abs x \in \AM$. By 
a  case 
analysis on $\sigma(x)$ and $\sigma(\abs x)$ we have:
$\sigma_s(x)\!=\!(V,I)$ and  $\sigma_s(\abs x)\!=\!(V,I)$, or $\sigma_s(x)\!=\!(V,V)$ 
and  
$\sigma_s(\abs x)\!=\!(V,I)$ or $(V,V)$.  Reading has no effect on validity status and in 
all 
cases we have $\sigma'(\abs x)~\leq~\sigma'(x)=\sigma_s(x)$.

\noindent$(2)$ $RW~\abs x \in \AM$ then similarly to the case (2) above we have 
$\sigma_s(\abs x)=(V,I)$, additionally $\sigma'(x)=\sigma_s(x)=(V,I)$ or $(V,V)$. In all 
cases $\sigma'(\abs x)\leq\sigma'(x)$.

\noindent$\bullet$ If  $x$ is not accessed but is in the declaration, the 
reasoning is the same as if it was only read. 
Note that a variable that is not accessed cannot be declared 
in write mode, $W~\abs x \in \AM$, by
Definition~\ref{def-WD}.
\end{proof}

Finally, a well-declared program always runs to completion: it never tries to access an 
invalid memory location.

\begin{Theorem}[Progress for well-declared programs]\label{thm-progress}
If a program $\Prog$ is well-declared, then its execution cannot reach a stuck 
configuration.
\end{Theorem}

\begin{proof}
 By 
Property~\ref{prop:stuck}, 
it is sufficient to prove that 
unification on the validity status is always possible. 
We consider a reduction  $(\transl{\AM\{S\}},\sigma) \to^* (S,\sigma_s) \to^* \ldots$ 
similarly  to the proof above.

By definition of well-declared 
programs and because of the signature of effects ($w~x$ cannot be stuck), only four
cases have to be analysed for the local operations:
 \begin{itemize}
\item \Pull\ operations (on $x$ and $\abs x$) in the translation of $R~\abs x$ or 
$RW~\abs x$. Unification 
requires that $\sigma(x)=(X,V)$ and  $\sigma(\abs x)=(Y,V)$.
\item $r~x$ operation in the evaluation of $S$. Unification 
requires that $\sigma'(x)=(V,X)$ where $\sigma'$ is the store in which the read access is 
to be evaluated.
\item $\Push~x$, $\Push~\abs x$, and $\rem{r~x}$ that are similar to the cases above, they require the symmetric validity status in similar conditions (not detailed below).
\end{itemize}
Indeed, access mode declarations do not generate reading operations, and by definition 
function statements contain no \Push\ or \Pull.

Concerning the first case, because of Theorem~\ref{thm-SR}, we have $\sigma(\abs x)\leq 
\sigma(x)$, and because of property~\ref{prop:safe} none of them is $(I,I)$. By case 
analysis on the possible values of $\sigma(\abs x)$ and $\sigma(x)$, it is easy to show 
that $\sigma(x)=(X,V)$ and  $\sigma(\abs x)=(Y,V)$ if we reach the two \Pull\ statements 
that perform data transfers before the execution of the function.

Concerning read access, they should be verified by an induction on the reduction steps 
following the state $(S,\sigma_s)$ showing that, for any variable $x$ that is declared $R$ 
or $RW$, in all states we have $\sigma'(x)=(V,X)$. Indeed, by the same analysis as in the 
proof of 
Theorem~\ref{thm-SR} we know that $\sigma'(x)=(V,X)$. Because of 
Property~\ref{prop-localised} no remote operation is possible on $x$ and thus only $r~x$ 
and $w~x$ operations are possible on $x$, both maintain the invariant $\sigma'(x)=(V,X)$ 
for some $X$.
%
\end{proof}
\begin{example}
Consider the example above of a variable written on the CPU, and then read on the GPU, 
a well-declared program encoding this behaviour would be:\\[2mm]
$\begin{array}{l}
RW~\abs x\{w~x\};\\
\rem{R~\abs x}\{\rem{r~x}\}
\end{array}$\\[1mm]
This code automatically generates the \Push\ instruction that prevents 
the program from being stuck, indeed the $RW$ annotation ensures that after executing the first line, the validity status pf $\abs x$ is: $(V,I)$, the encoding of  for $\rem{R~\abs x}$ checks whether the remote status is valid, as it is not the statements: $\Push~x;\Push~\abs x;$ are executed.
\end{example}

\subsection{Effects and Access Mode Declarations for Arrays}\label{sec-arrays}
In array structures, the 
validity status of the whole array is abstracted away by a single validity status pair. 
We extend the syntax for arrays as follows, $x[i]$ denotes the indexed access to an 
element of the array. More precisely the new operations on arrays and their elements are 
(we still have the previous operations on non-array and abstract variables):
\[S::= ... \,|\, r~x[i] \,|\, w~x[i] \]

Synchronisation operations (\Push\ and \Pull) exist for arrays but the whole array is 
synchronised, and we write $\Push~x$ and $\Pull~x$ as above.
All the elements of the array are represented by a single abstract variable: $\abs{x}$ 
represents a safe abstraction of the validity status of all $x[i]$. In other words, as soon as one element of the array $x$ is invalid locally (resp. remotely) the validity status of $\abs x$ can only be $(I,V)$ (resp. $(V,I)$).

The semantics of access mode declarations and programs is unchanged because 
synchronisation operations and access mode declarations do not concern array elements.
The concept of well-protected programs must be adapted to the case of array structures, 
and more 
precisely to the fact that several  memory locations are represented by a single abstract 
variable.

\begin{definition}[Well-declared program with array access]\label{def:well-declared-array}
A program $\Prog$ is \emph{well-declared} if for all $\AM\{S\}$ in $\Prog$, additionally 
to the rules of Definition~\ref{def-WD}, we have\footnote{$\range(x)$ returns the set of valid indices of the array $x$}:
\begin{itemize}
\item $w\ x[i]\in S \implies (W\ \abs x \in \AM \lor RW\ \abs x \in \AM)$,
\item $r\ x[i]\in S \implies (R\ \abs x \in \AM \lor RW\ \abs x \in \AM)$,
\item $W\ \abs x\!\in\! \AM \!\implies\! \forall i\!\!\in\!\!\range(x).\, w\ x[i]$ occurs in all execution 
paths of $S$,
\item Plus the same rules for remote operations.
\end{itemize}
\end{definition}

\begin{example}
Consider the function body $\rem {w~x[3]}$, corresponding to a function made of the C++ statement $x[3]=0$ executed on a GPU. The only safe access mode declaration for it is $\rem{RW~\abs x}$ indeed, some of the elements of the array are written but not all. At the end of the function execution the valid copy of the array is on the GPU.
\end{example}

The other definitions and properties are expressed similarly for arrays, compared to standard variables, and both \emph{subject-reduction}, 
Theorem~\ref{thm-SR}, and \emph{progress}, Theorem~\ref{thm-progress}, are 
still valid. The only change is the ``correct abstraction of the memory state'' criteria -- see Definition~\ref{CorrectAbstraction} --  
that becomes 
\[\forall x\in 
\vars(S).\begin{array}[t]{l}
\sigma(\abs x)\leq\sigma(x) \text{ if $x$ is not an array} \\
	\forall i\in\range(x).\, \sigma(\abs x)\leq\sigma(x[i]) \text{ if $x$ is an array}
\end{array}
\]

%
%
\begin{Theorem}\label{thm-correct-array}
If a program using arrays is well-declared according to Definition~\ref{def:well-declared-array} then its execution verifies both the subject-reduction and the progress property.
\end{Theorem}
\begin{proof}[Proof sketch]
 The proofs are similar to the non-array case except in the case of $W~x$ 
declarations where the fact that all elements of the array must be written is necessary 
to ensure that no element is in the status $(I,V)$ (which could not be safely represented 
by $(V,I)$) at the end of the function execution. If we focus on the proof of 
Theorem~\ref{thm-SR}, case ``$x$ is written, sub-case (1) we have $\sigma'(\abs x)=(V,I)$ 
which is a safe abstraction because \emph{all elements have been written}, and thus 
$\sigma'(x[i])=(V,I)$ for all $i$. If one element $j$ was not written, we could have had 
$\sigma'(x[i])=(I,V)$ which would invalidate the theorem. Overall, only arguments about correctness of the abstractions need to be adapted.
\end{proof}

\subsection{Discussion: Similarities and differences relatively to VectorPU}

Let us compare the formal definition of the coherence protocol, Figure~\ref{sem-AM} (valid for simple memory locations or arrays), with the VectorPU implementation of the protocol for simple arrays, Figure~\ref{fig:vectorpucoherence}. Except  the order of operations and minor changes, the code is similar. The main difference is that there is no view of abstract vs. concrete variables, however, if we consider that transfer operations on abstract variables have no effect, and that validity status of concrete variables can be abstracted away, the code is  the same as the formalisation. 

Taking a more global point of view, no verification is performed by the VectorPU framework and thus, the property of ``well-declared programs'' is not checked currently by the framework. In the current state of the library, the property ``well-declared programs''  must be ensured by the programmer. The implementation of VectorPU relies on the hypothesis that the function declarations are correct, because of this the current formalisation is a significant step forward as it allows us to express precisely what assumption is made by the library on the programmer's code.

Such a check could either be done  by runtime verification checking that each function performs exactly the required access, or  statically by constraints on the program and a static analysis (involving some approximations, meaning  some correct programs could be rejected). The first solution  is not acceptable considering the target application domain because of the overhead involved by the dynamic checks. 
Let us now try to figure out how difficult it would be to ensure that a given program verifies the ``well-declared'' property statically. We focus on the case of arrays that is the most interesting. Checking which accesses are performed on an array is a difficult task in general. However for reading access, declaring the  array of constant type could enforce a reading access and even allow us to infer the $R$ annotation for constant arrays. $RW$ is not constraining the access and is always a safe annotation. $W$ however requires all the elements to be written, checking this relies on a static analysis that is way beyond the scope of this paper and might be tricky. As already argued, the fact that the same variable does not appear twice is more  a restriction of the formalisation provided the same array is always declared with the same access ($RW$ in the worst case), an additional pass could modify the annotations of identical variables so that the least restrictive one is chosen. Possible aliasing between variables can be seen as a particular case of overlapping array, studied in the next section. 

Finally concerning expressiveness, the ``well-declared programs'' definition is a bit more flexible than what VectorPU targets at the moment because it is safe to declare in our framework functions that access some variables on the CPU and others on the GPU, and this is not planned in VectorPU, again due to the supported usage scenario.

These differences highlight interesting improvement directions for VectorPU while we can still consider that the current article  is a faithful formalisation of the library. In the next section, we investigate a feature that is not yet supported by VectorPU, but exists in SkePU. It concerns the handling of arrays that may overlap; this generalizes the problem of  aliasing between variables.
 We can somehow consider the theoretical results below as a specification of a future extension of the library.

\section{Overlapping arrays}\label{sec:overlap-array}

In the preceding section we supposed that arrays were well-separated. The preceding abstraction would also be valid for an array that would be split into disjoint entities (such as
VectorPU \texttt{pvector}s) and always used either as the disjoint sub-arrays or as the whole. In this section we extend the framework to take into account overlapping arrays. Here we still consider single dimension arrays for simplicity but multi-dimensional arrays could easily be taken into account.

\subsection{Context and Objectives}

In VectorPU, the first \texttt{pvector} on a vector passed as argument for access on device will (over)allocate space for the entire \texttt{vector}, all subsequent \texttt{pvector} accesses to the same \texttt{vector} can skip the allocation.
 Consequently, two successive data transfers of the same memory location will be written to the same memory location, even if the two initial locations are accessed through different (overlapping) \texttt{pvector}s.
Consequently, on the formal side, if several push or pull are performed on overlapping memory locations, the transferred memory has the same overlaps as the source. This is very important to ensure that no two copies of the same array element will coexist in the same location.

\paragraph{Problem statement}
Overlapping arrays raises several difficulties making the approach currently adopted in VectorPU not adapted. Indeed a single write operation can change the validity status of a cell that belongs to several arrays. Consequently, several access annotations may have to be written for a single operation. As a consequence, some annotations may never be correct: a function that is declared to write all the elements of an array $x$ will necessarily write some elements of the arrays that have an overlap with $x$, these overlapping arrays should thus be annotated and transmitted, or another coherence protocol should be used.

\paragraph{Approach}
In this work we take the decision not to change the coherence protocol of VectorPU and instead work at the access mode declaration level to ensure the consistency of overlapping arrays.

To take into account overlapping arrays in VectorPU, two approaches could be envisioned. A naive solution consists in applying the results for non-overlapping arrays. Indeed, Definition~\ref{def:well-declared-array} of well-declared programs with array accesses is still valid. However,  the programmer now has to annotate more variables because each array access operation may involve several arrays. Due to the array overlap, the programmer should now know all the arrays that are impacted by a function execution, including the overlapping arrays that are not passed as parameter, and the library should be extended to pass them as ``artificial'' parameters. To be more explicit, Definition~\ref{def:well-declared-array} should be extended as follows (with a symmetrical rule for remote writing):

\begin{quote}
For every $y$ overlapping $x$, if $i$ is in the range in common between $x$ and $y$, we have $w\ x[i]\in S \implies (W\ \abs y \in \AM \lor RW\ \abs y \in \AM)$.
\end{quote}

First note that no additional rule is necessary for read operations. Indeed, read operations use the validity status but do not modify it, consequently, read operations do not modify the validity status of other arrays and have no consequence on overlapping arrays.
Note also that the above rule restricts a bit the expressible effects as, for example, $W\ \abs x ; \rem{R\ \abs y}$ cannot be valid if $x$ and $y$ overlap. Indeed, $W\ \abs x$ means that all the cells of $x$ are written and thus the access mode for $\abs y$ must be $W$ or $RW$. In this case, this is indeed a safe restriction as the array $x$ will be written, and reading $y$ remotely might not access a valid value.

%

\subsection{Access mode inference for overlapping arrays}\label{sec:infer-overlap}

Though very precise, the approach described above does not seem realistic and we additionally develop an inference mechanism for access mode declarations in presence of overlapping arrays.
The objective  is to infer the correct annotations on  variables that are not passed as parameters. Knowing the access modes for the function parameters, we infer what operation must be done on other intersecting arrays to ensure the coherency of the system, and we express these additional operations as implicit generated annotations. These additional access mode declarations are inferred, the semantics of these added declarations will result in additional data transfers and validation/invalidation operations that make the program correct. This is less precise as it takes a pessimistic approach on the operations performed by the declared arrays. For example, for any array that is declared $RW$ we will suppose that all the array elements may be read and written, but the approach is safe and mostly automatic.

Starting from a given set of access mode annotations, we want to infer other access modes that are consequences of the overlaps and the existing annotations.
Because we are only aware of an approximation of the effects (for each array variable, effect is abstracted by a single global effect), the inferred accesses will be approximate but can be a safe over-approximation of the effect of the function.
 Without knowing the real accesses performed by the function, we deduce from the declared access modes, a set of additional ``artificial accesses''.

We consider $\AM$ the set of all access modes declared for a given function and extend it so that the function satisfies the well-declared program requirement even with overlapping arrays.

\begin{table}[!tb]
\begin{mathpar}
\inferrule
{E\ \abs x \in\AM}
{E\ \abs x \in\Overlap \AM}

\inferrule
{\rem{E\ \abs x} \in\AM}
{\rem{E\ \abs x} \in\Overlap \AM}

\inferrule
{W\ \abs x \in\AM \\ W\ \abs y \not\in\AM \\ x \text{ and } y \text{ overlap}}
{RW\  \abs y \in\Overlap \AM}

\inferrule
{RW\ \abs x \in\AM \\ x \text{ and } y \text{ overlap}}
{RW\  \abs y \in\Overlap \AM}

\inferrule
{\rem{W\ \abs x} \in\AM \\ \rem{W\ \abs y} \not\in\AM \\ x \text{ and } y \text{ overlap}}
{\rem{RW\  \abs y} \in\Overlap \AM}

\inferrule
{\rem{RW\ \abs x} \in\AM \\ x \text{ and } y \text{ overlap}}
{\rem{RW\  \abs y} \in\Overlap \AM}

\end{mathpar}
\caption{Extension of access mode annotations to deal with overlapping arrays}\label{Overlap-ext-tab}
\end{table}
\begin{example}
To understand the principle of the approach, consider the case where 
$W\ \abs x\in \AM$ then $ \forall i\in\range(x). w\ x[i]$ occurs in all execution 
paths of $S$, and thus for all arrays $y$ overlapping $x$  we must have $(W\ \abs y \in \AM \lor RW\ \abs y \in \AM)$. If we have no additional information we will ensure that $RW\ \abs y$ is also in the set of access mode declarations, which is always safe.
The extension $\Overlap \AM$ defined below not only contains the original annotations $\AM$ but also the annotations required for coherency of overlapping array. Here if $W\ \abs x\in \AM$ then both $W\ \abs x$ and $RW\ \abs y$ are in  $\Overlap \AM$.
\end{example}

\begin{definition}[Extension of access mode declarations for overlapping arrays]\label{def-overlap-annotation}
Consider a set $\AM$ of access mode declarations.
 The extension for overlapping arrays of $\AM$  is the smallest set $\Overlap \AM$ defined by the  rules in Table~\ref{Overlap-ext-tab}.

\end{definition}

The following theorem states that if the set of function parameters is extended according to the preceding definition, then memory consistency is ensured. Note that it means that a set of artificial parameters are to be added to some functions, in the sense that  data-transfers and validity status modifications  have to be performed on vectors that are not among the original parameters of the function.

\begin{Theorem}\label{thm-correct-array-overlap}
Consider a program that is well-declared according to Definition~\ref{def:well-declared-array}, not taking into account  overlapping arrays. Suppose its access mode declarations are extended according to Definition~\ref{def-overlap-annotation} then the execution of the obtained program verifies both subject reduction and progress, even in presence of overlapping arrays.
\end{Theorem}

\begin{proof}[Proof sketch]
The principle of the proof is to prove that, provided a set $\AM$ correctly declares the accesses performed syntactically by a function on its parameters, the set $\Overlap \AM$ is a correct approximation of the accesses performed by the function on all the arrays of the program, i.e. the parameter arrays and the arrays that overlap the parameter arrays. Then Theorem~\ref{thm-correct-array} will be  sufficient to conclude.

Trivially, the first two rules of Table~\ref{Overlap-ext-tab} are sufficient to conclude about normal function parameters. We now need to ensure that operations on overlapping arrays are well-declared. We focus on non-remote operations and prove that (indirect) operations on overlapping arrays are well-declared, according to Definition~\ref{def:well-declared-array} modified by the additional rule introduced in above: 
\begin{quote}
For every $y$ overlapping $x$, if $i$ is in the range in common between $x$ and $y$, we have $w\ x[i]\in S \implies (W\ \abs y \in \AM \lor RW\ \abs y \in \AM)$.
\end{quote}
By a simple case analysis on the possible annotations and the possible operations performed on the arrays, we deduce that the access modes added by the four last rules of Table~\ref{Overlap-ext-tab} are sufficient to ensure the statement above. Finally the symmetrical statement for remote writing is ensured in the same way.
\end{proof}

\subsection{Towards an implementation in VectorPU}

To implement the proposed mechanism that ensure the safety of overlapping vector accesses
in a function call $f$, we need to add to VectorPU the two following components:
\begin{itemize}
    \item A representation $r_v$  to retrieve, for any given \texttt{pvector} $pv$
             of a \texttt{vector} $v$, the set of all other valid \texttt{pvector}s of $v$ 
             that overlap with $pv$.
             $r_v$ is initialized as empty when declaring a new \texttt{vector} $v$,
             queried and/or updated at \texttt{pvector} creations, deletions, and at calls,
             and is removed when $v$ is deallocated.
    \item A mechanism which intercepts the function call $f$ and,
          for every vector operand (\texttt{vector} or \texttt{pvector}) $pv$ 
          accessed as W or RW in $f$,
          looks up in the corresponding representation $r_v$ all  
          \texttt{pvector}s overlapping with $pv$. 
          For other arguments in $f$ that overlap with $pv$ and have access mode
          R or W, their access mode
          is updated to RW, as proposed in Table~\ref{Overlap-ext-tab}.
          For any other existing \texttt{pvector}s $w$ of $v$ not accessed in $f$
          (but possibly in earlier and/or later calls) 
          that do overlap with $pv$, 
          we append shadow arguments $RW(w)$ to $f$ as suggested by Table~\ref{Overlap-ext-tab}.
          Finally, the intercept mechanism performs, as before, the resulting coherence
          actions (data transfers, status updates) and delivers the call.
          For intercepting the function call, the function call operator is overloaded.
\end{itemize}

\begin{example}
As a simple example, let us consider the following set of \texttt{pvector}s and call sequence:

\begin{verbatim}
v = new vector(10, ...);
pv1 = new pvector( v, [2:5] );
pv2 = new pvector( v, [4:8] );
pv3 = new pvector( v, [7:9] );
pv4 = new pvector( v, [2:3] );
...
f1( ... R(pv1), R(pv2), ... );
f2( ... W(pv3) ... );  
f3( ... RW(pv4), R(pv2), ... );
\end{verbatim}

Intercepting the function calls, we maintain $r_v$ and update the calls as 
described above.
For the call to \verb.f2., we infer from W(pv3) and the overlap of \verb.pv3. with \verb.pv2.
by Table~\ref{Overlap-ext-tab} that the access mode of \verb.pv2. (not accessed in \verb.f2.) must be 
upgraded to RW, which we do by conceptually appending \verb.RW(pv2).
as a shadow argument to \verb.f2.. The call to \verb.f2. is thus
conceptually rewritten\footnote{As all overlapping
\texttt{pvectors} had been identified before the rewriting, the rule needs
not be applied recursively to the appended shadow arguments, here \texttt{RW(pv2)}.} into

\begin{verbatim}
f2( ... W(pv3) ... , RW(pv2) ); 
\end{verbatim}

\noindent
hence we make sure that the access to \verb.pv2. in the subsequent
call to \verb.f3. will be handled correctly.
\end{example} 

It remains to select an appropriate data structure for $r_v$ that allows for
efficient dynamic insertion and removal of \texttt{pvector}s of a vector $v$, 
i.e., index intervals, and efficient lookup of all pvectors that overlap
with a given query interval.
For very small numbers of \texttt{pvector}s of a vector $v$, a simple unordered
list of \texttt{pvector}s is sufficient; this is used e.g.\ in the
smart-container coherence management in SkePU \cite{Dastgeer-IJPP15}. 
For scaling up to larger numbers of \texttt{pvector}s,
a \emph{segment tree} \cite[Sec.~10.3]{Overmars} 
could be used. A segment tree storing $n$ intervals
can be updated dynamically (insertion, removal) in time  $O(\log n)$ and
can retrieve the set of all $k$ intervals overlapping with a query interval
in time $O(k+\log n)$; the space requirements is $O(n\log n)$.

The time and space required to handle overlapping arrays at runtime could be saved with a precise enough static analysis that would infer potential overlapping between arrays; choosing a precise enough or static analysis designing a dedicated one is outside the scope of this paper.

%

\section{A few related works}\label{sec:RW} 
Most of the verification works related to memory consistency focus on coherence 
protocols 
and/or 
weak memory models~\cite{pong1997verification}. 
Among them, one could cite works on a lazy caching algorithm~\cite{Gerth1999}, a formal specification of a caching 
algorithm, and its verification in TLA~\cite{Ladkin1999}. These works shows the 
difficulty to reason on memory coherency, but also that specifications in these models 
should rely on a few simple instructions on the type of memory accessed, a bit similarly 
to this proposal.
Coherence protocols have also been verified using CCS specifications~\cite{Barrio01}. 
These various works are quite different from the approach presented in this paper because 
we rely here on a declarative approach for memory accesses: the programmer declares the 
kind of memory accesses performed by a component, and the consistency mechanism ensures 
that each component accesses a valid memory space.

More recently, and adopting a more language oriented approach, Crary and 
Sullivan~\cite{Crary:POPL:2015} designed 
a calculus for expressing ordering of memory accesses in weak memory models. We 
are interested here in a much simpler problem where memory access is  sequential 
and clearly identified but the objective is to prove the correctness of simple cache-coherency operations. 
Even an  extension of this work for parallel processes would  result in a 
simpler model than the ones that exist for weak memory models because of the explicit 
consistency points introduced in the execution by the start/end of each function. 

The closest work to ours is probably~\cite{BJPTSAC16} that defines a memory access calculus similar to ours and proves the 
correctness of a generic cache coherence protocol expressed as part of the semantics of 
the calculus. Compared to this work, we are interested in explicit statements on memory 
accesses and thus the cache consistency is partially ensured  by the programmer 
annotations, making the approach and the properties proven significantly different. Some 
aspects of the approaches could have been made more similar, e.g. by extending our 
work to more than two address spaces or adopting a different  syntax. However our problem 
and formalisation are quite simpler, and we 
believe easier to read, while sufficient for our study.
The same authors also designed a formal model written in Maude~\cite{BJPTMaude16} to 
better understand the possible 
optimisations and the impact of the memory organisation on performance in the context of 
cache coherent multicore architectures. This could be an interesting starting point for 
future works, especially if we 
extend our work to better model the performance aspects of VectorPU and want to reason 
formally on the improved performance obtained by the library. Also from the same authors~\cite{Bijo2017} extends the results described above with parallel spawned task and could be a source of inspiration  to extend our work towards parallel function execution.

\section{Conclusion and future works}\label{sec:conclusion}
In this article we provided a formal approach to verify the consistency of the memory 
accesses in heterogeneous computer systems made of two memory spaces. We formalise the 
operations of memory accesses and memory synchronisation between the two memory spaces 
and prove that a program adequately annotated with informations on the memory accesses 
always access valid memory spaces and tracks correctly which of the memory space contains 
the up-to-date data.

The practical result is that we can verify the coherency mechanism used by the VectorPU 
library and ensure that, additionally to the significant performance benefits of the 
approach, the VectorPU mechanisms is correct and ensures the consistency of the memory 
accesses.

We also extended our model for studying the effect of
operations made on overlapping arrays.
The current implementation of VectorPU supposes that the 
(\texttt{pvector}) array operands always
represent disjoint memory locations, it does not 
take into account overlapping arrays. 
Based on the solution developed in our model,
we described an extension of the VectorPU library that could 
deal safely with overlapping array accesses by
overlapping \texttt{pvector} arguments.

We envision several extensions to this work.
The current article only deals with two memory spaces; the extension to many 
memory spaces (as supported e.g.\ in SkePU)
seems relatively simple but  the mechanism dealing 
with memory transfers between several memory locations becomes a bit more complex; its 
formalisation should be similar.

Moreover, we are interested in the application of our approach to the 
verification of other frameworks. 
Indeed VectorPU uses the most primitive
cache coherence protocol, the VI-protocol.
 More elaborated coherence protocols like MSI or MESI 
 (as used e.g.\ in SkePU \cite{Dastgeer-IJPP15}) 
 introduce additional states where the
 number of readers has to be tracked for example. 
 Also, SkePU uses a more space-efficient management of
 partial vector accesses, the
 coherence protocol itself involves explicit intersection tests with
 existing copies. 
 Verifying such framework would require 
 a modification of our abstract state representation and a modification of the access 
 mode translational semantics.
 
\section*{Acknowledgments}   

C.\ Kessler acknowledges partial funding by
EU H2020 project EXA2PRO (801015).

\section*{References}
\bibliographystyle{elsarticle-num} 
\bibliography{bibliography}

\end{document}

%% file: vectorpu-intro.tex
Heterogeneous computer systems, such as traditional CPU-GPU based systems, 
often expose disjoint memory spaces to the programmer,
such as main memory and device memory, with the need
to explicitly transfer data between these.
The different memories usually 
require different memory access operations and
different pointer types. 
Also, encoding memory transfers as message passing communications
leads to low-level code that is more error-prone. 
A commonly used software technique to abstract away the
distributed memory, the explicit message passing,
and the asymmetric memory access mechanisms
consists in providing the programmer with an 
object-based shared memory emulation. For CPU-GPU systems,
this can be done in the form of special data-containers,
which are generic, STL-like data abstractions such as 
\verb+vector<...>+ that 
wrap multi-element data structures such as arrays. These 
data-container objects
 internally perform transparent, coherent
software caching
of (subsets of) accessed elements in the different memories 
so they can be reused (as long as not invalidated) 
in order to
avoid unnecessary data transfers. Such data-containers 
(sometimes also referred to as ''smart'' containers as they can
transparently perform data transfer and memory allocation optimizations
\cite{Dastgeer-IJPP15}) 
are provided in a number of programming frameworks 
for heterogeneous systems, such as
StarPU \cite{StarPU} and SkePU~\cite{Enmyren10,Dastgeer-IJPP15}. StarPU is a 
C-based library that provides API functions to
define multi-variant tasks for dynamic scheduling
where the data containers are used for modeling 
the operand data-flow
among the dynamically scheduled tasks. 
SkePU defines device-independent 
multi-backend skeletons like map, reduce,
scan, stencil etc.\ where operands are
passed to skeleton calls within data containers.

VectorPU \cite{VectorPU-2017} is a recent C++-only
open-source
programming framework for CPU-GPU heterogeneous systems. 
VectorPU relies on the specification of 
\textit{components}, which are functions that contain kernels for execution
on either CPU or GPU. Programming in VectorPU is thus not restricted 
to using predefined skeletons like SkePU, 
but leads to more high-level and more concise code than StarPU. 
Like StarPU, VectorPU requires the programmer
to annotate each operand of a component
with the access mode (read, write, or both) including the 
accessing unit (CPU, GPU), and uses smart data containers for automatic transparent
software caching based on this access mode information.

The implementation of VectorPU makes excessive use of static metaprogramming; this provides a light-weight realization of the access mode annotations and of the software caching, 
which only require a standard C++ compiler. Emulating these 
light-weight
component and access mode constructs without additional language
and compiler support (in contrast to, e.g., OpenACC or OpenMP), 
leads however to some compromises concerning the possibility to perform static analysis.
In particular, VectorPU has no explicit type system for the
access modes, as these are not known to the C++ compiler.

In this paper, we  formalize access modes
and data transfers in CPU-GPU heterogeneous systems and prove 
the correctness of the software
cache coherence mechanism used in VectorPU.
The contributions of this paper are:

\begin{itemize}
\item A simple effect system modeling the semantics of memory
   accesses and communication in a CPU-GPU heterogeneous system,
\item A small calculus expressing different memory
   accesses and their composition across program traces. 
\item The interpretation of VectorPU operations as higher-level statements
    that can be translated into the core calculus,
\item A proof that, if all memory accesses are performed 
    through VectorPU operations, the memory cannot reach an 
    inconsistent state and all memory accesses succeed,
\item The abstraction necessary to take into account arrays, possibly overlapping, in the formalism.
\end{itemize}

This article is an extended version of \cite{HKL-4PAD2018}, with two main additions. First the relationship between the formal results and the VectorPU implementation is  detailed, illustrating the impact of the proven results on the behaviour of the library. Second, the theoretical framework is extended to take into account the fact that manipulated arrays may overlap and that the consistency mechanism must take this information into account to be correct. While overlapping arrays are not yet supported by VectorPU, based on the formal model we develop, we show how a simple extension of the library could provide support for overlapping arrays.

This paper is organized as follows:
Section~\ref{VectorPU} reviews VectorPU as far as required for  this paper, for further information we refer to \cite{VectorPU-2017}.
Section~\ref{sec:Formal} provides our formalization of VectorPU programs and
their semantics, and proves that the coherence mechanism used in
VectorPU is sound. Section~\ref{sec:RW} discusses related work, and 
Section~\ref{sec:conclusion} concludes.

\begin{figure}
\begin{center}
\includegraphics[width=0.64\textwidth]{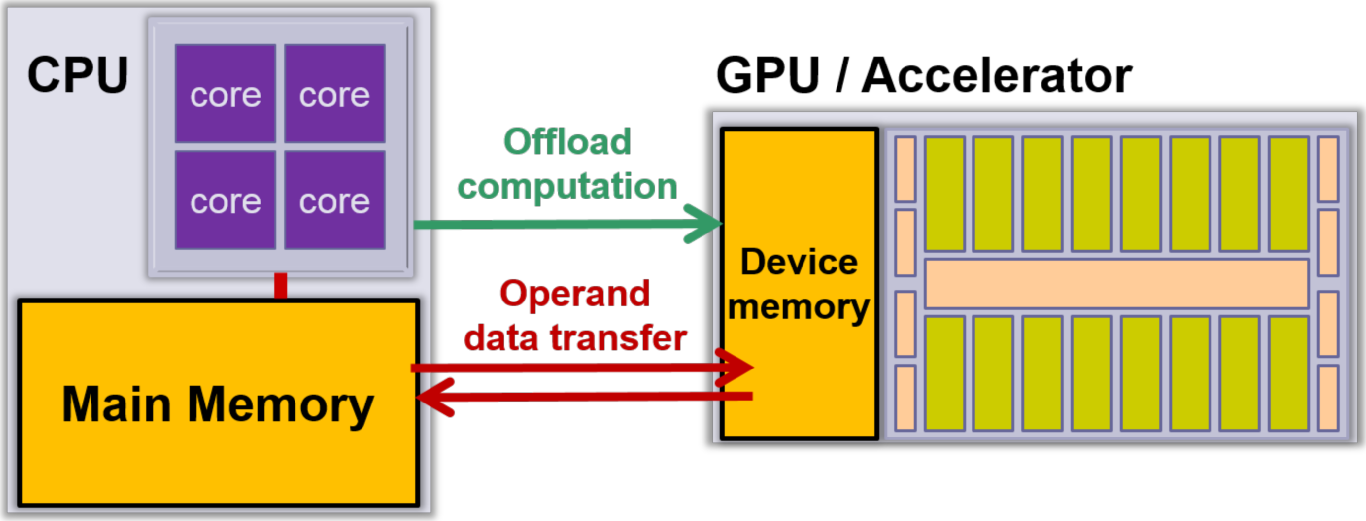}
\caption{\label{fig:CPU-GPU}A GPU-based system with distributed address space}
\end{center}
\end{figure}

\section{VectorPU}\label{VectorPU}

In heterogeneous systems with separate address spaces, for example in 
many GPU-based systems, a general-purpose processor
(CPU) with direct access to main memory is connected by some network
(e.g., PCIe bus) to one or several accelerators
(e.g., GPUs) each having its own device memory,
see Figure~\ref{fig:CPU-GPU}. Native programming models
for such systems such as CUDA typically expose the distributed address spaces to the programmer, who has to write explicit 
code for data transfers and device memory management.
Often, programs for such systems  must be organized in multiple source files
as different programming models and different toolchains
are to be used for different types of execution unit.
This enforces a low-level programming style. 
Accordingly, a number of single-source 
programming approaches have been proposed that abstract away the distribution 
by providing a virtual shared address space. Examples include directive-based
language extensions such as OpenACC and OpenMP4.5, and C++-only approaches 
such as the library-based skeleton programming framework SkePU \cite{Dastgeer-IJPP15}
and the recent macro-based framework \textit{VectorPU}.

\textit{VectorPU} \cite{VectorPU-2017} is an open-source\footnote{http://www.ida.liu.se/labs/pelab/vectorpu, https://github.com/lilu09/vectorpu} lightweight C++-only 
high-level programming layer
for writing single-source heterogeneous programs for Nvidia CUDA GPU-based systems.
Aggregate operand data structures 
passed into or out of  function calls are
to be wrapped by special data containers known to VectorPU.
VectorPU currently provides one generic data container,
called \verb+vector<...>+,
with multiple variants that 
eliminate the overhead of managing heterogeneity and distribution when not required (e.g., when no GPU is available). 
\verb+vector<...>+ inherits functionality from STL \verb.vector. 
and from Nvidia Thrust \verb.vector., and
wraps a C++ array allocated in main memory. 
VectorPU automatically creates on demand
copies of to-be accessed elements in device memory and keeps all copies coherent using
a simple coherence protocol, data transfers are only performed when needed. 

VectorPU programs are organized as a set of C++ functions, some of which
might internally use device-specific programming CUDA constructs\footnote{%
VectorPU allows to directly annotate a CUDA kernel function, in addition to annotating its C++ wrapper function.} while others
are expected to execute on the host, using one or possibly multiple cores.
VectorPU \emph{components} are functions that are supposed to contain (CPU or device)
kernel functionality and for which  operands are passed as VectorPU data container objects. 
Components and the types of execution units that
access their operands are declared 
by annotating the operands of the function, either at a call of the function
or for the formal parameters in the function's declaration, 
with VectorPU \emph{access mode specifiers}. For example, in contrast,  SkePU \cite{Enmyren10}  overloads element
access and iterator operations so that monitored 
accesses are also possible on demand in non-componentized (i.e., 
ordinary C++) CPU code.
VectorPU only relies on access mode annotations 
to perform lazy data transfer,
not knowing when data is going to be accessed inside a component.

Table~\ref{tab:modes} summarizes the access mode annotations
currently defined for VectorPU. The access mode specifiers,
such as \texttt{R} (read on CPU), \texttt{W} (write on CPU), \texttt{RW} 
(update, i.e., both read and write, on CPU), \texttt{GR} (read on GPU) and so forth,
are available both as annotations of function signatures and
as C++ preprocessor macros that expand at compilation into (possibly, device-specific) C++ pointer
expressions and side effects that allow to generate device specific access code
and use device-specific pointer types for the chosen execution unit. 
For instance, \texttt{GW(x)} expands to a GPU pointer to
the GPU device copy of \texttt{x},
which might be dereferenced for GPU writing accesses to \texttt{x},
such as the GPU code: \verb:*( GW(x) + 2 ) = 3.14:.
\texttt{GWI(x)} evaluates to a Thrust-compatible iterator onto the 
GPU device copy of \texttt{x}, and \texttt{WEI(x)} to an iterator-end reference
to the last element of \texttt{x} on CPU side. The current VectorPU prototype implementation does not
(yet) check access-mode annotations in signatures of externally defined functions.
It is also possible to specify partial access of a \verb:vector:
instead of the
entire \verb.vector. data structure. The current
VectorPU implementation does not (yet) support coherence for 
\emph{overlapping}
intervals of elements resulting from multiple (partial) accesses
some of which (may) access the same element.
A solution for this problem has been described for SkePU
smart containers by Dastgeer~\cite{Dastgeer-IJPP15}. Section~\ref{sec:overlap-array} details a solution for handling overlapping arrays in VectorPU.

\begin{table}
\caption{\label{tab:modes}VectorPU access mode annotations for a parameter  \cite{VectorPU-2017}}

\begin{center}
\begin{tabular}{|lll|}
\hline
Access Mode & On Host & On Device \\
\hline
Read pointer & \texttt{R} & \texttt{GR} \\
Write pointer & \texttt{W} & \texttt{GW} \\
Read and Write pointer & \texttt{RW} & \texttt{GRW} \\
Read Iterator & \texttt{RI} & \texttt{GRI} \\
Read End Iterator & \texttt{REI} & \texttt{GREI}\\
Write Iterator & \texttt{WI} & \texttt{GWI}\\
Write End Iterator & \texttt{WEI} & \texttt{GWEI}\\
Read and Write Iterator & \texttt{RWI} & \texttt{GRWI}\\
Read and Write End Iterator & \texttt{RWEI} & \texttt{GRWEI}\\
Not Applicable & \texttt{NA} & \texttt{NA} \\
\hline
\end{tabular}
\vspace{-3ex}
\end{center}
\end{table}

The following example (adapted from \cite{VectorPU-2017}) 
of a CUDA kernel wrapped in an 
 annotated function \verb.bar. shows the use of 
VectorPU access mode annotations at function declaration:

{\footnotesize \begin{verbatim}
// Example (annotations at function declaration): 
__global__
void bar ( const float *x [[GR]], float *y [[GW]],
                 float *z [[GRW]], int size )
{ ... CUDA kernel code ... }
\end{verbatim}}

Here, the operand array pointed to by \verb.x. may be read (only) by the GPU within \verb.bar.,
operand array \verb.y. must be written (only) by the GPU, and 
operand array \verb.z. may be read and/or written by the GPU.
When calling \verb.bar., the
first three operands are  passed as VectorPU \verb.vector. 
container objects.
The \verb.size. formal parameter is a scalar (not a data container), so it
will be available on GPU on a copy-in basis but no coherence will
be provided for it by VectorPU.

It is also possible to put the annotations into a call, and hence characterize a function as a VectorPU component:

{\footnotesize \begin{verbatim}
// declare a CPU function:
void foo ( const float *x, float *y, float *z, int size );

// declare three vectors:
vectorpu::vector<float> vx(100), vy(100), vz(100);

// call to VectorPU annotated function foo:
foo ( R( vx ), W( vy ), RW( vz ), size ) ;
\end{verbatim}}

Here, the access mode specifiers and the resulting coherence policy
 only apply to that particular invocation of \verb-foo-, while other
invocations of \verb-foo- might use different access mode specifiers.

The following example shows how to use iterators:

{\footnotesize \begin{verbatim}
vectorpu::vector<My_Type> vx(N);
std::generate( WI(vx), WEI(vx), RandomNumber );
thrust::sort( GRWI(vx), GRWEI(vx));
std::copy( RI(vx), REI(vx), ostream_iterator<My_Type>(cout, ""));
\end{verbatim}}

\noindent 
where \verb.std::generate. is a CPU function filling a section between
two addresses with values (here, random numbers),
and \verb.thrust::sort. denotes the GPU sorting functionality
provided by the Nvidia Thrust library. 



\subsection{Partial Vectors}
\label{sec:Partial Vectors}


Using iterators, it is possible in VectorPU to define references to
contiguous subranges of a vector,
called \textit{partial vectors} (\texttt{pvector}s), 
which can be passed as \texttt{vector}-compatible
operands to a function instead of an entire \texttt{vector}.
In this way, it is possible to pass several (disjoint or even
overlapping) \texttt{pvector} objects as seemingly different \texttt{vector}
arguments that however are just windows onto
a common \texttt{vector} container variable. 
In contrast to \texttt{vector}s without \texttt{pvector}s, where coherence is 
managed automatically by VectorPU, the coherence 
management in the presence of \texttt{pvector}s 
is exposed to the programmer.

A partial vector can be initialized by two iterators to a normal VectorPU \texttt{vector} (we call it \emph{mother vector}).
No new memory is allocated for this partial vector, 
only the two iterators are stored, 
and the coherence state for its range in the mother vector.
When a partial vector is declared, it automatically
inherits the coherence state information
from its mother vector. 

\begin{figure}
\noindent 
\begin{minipage}{\linewidth}
\begin{footnotesize}
\begin{verbatim}
struct my_set {
   template <class T>
   __host__ __device__
      void operator() (T &x) { x+=101; } 
};
vectorpu::vector<int> vx(10);  // the mother vector
vectorpu::pvector<int> vy(x, vx.begin(), vx.begin()+2);
vectorpu::for_each<int>( GWI(vy), GWEI(vy), my_set() );
vectorpu::for_each<int>( GWI(vy), GWEI(vy), my_set() );
SR( vy );  // explicit coherence management
vectorpu::for_each<int>( RI(vx), REI(vx), [](auto x) {cout<<x<<" ";} );  
\end{verbatim}
\end{footnotesize}
\end{minipage}
\caption{\label{fig:pvector}Example of using a partial vector (\texttt{pvector}) and the \texttt{SR} macro for explicit coherence management. (Note: \texttt{pvector} is a short form and actually called \texttt{parco\_vector} in the VectorPU API.)}
\end{figure}

The aliasing introduced by \texttt{pvectors} can lead to coherence problems. One such scenario could be that the
programmer intends to operate on the previous \texttt{vector} again after
some part of it was updated via a \texttt{pvector}, 
hence the whole \texttt{vector} would be in an inconsistent state.
In such cases, VectorPU expects the programmer to use a macro \texttt{S$X$} (\texttt{S}ynchronize for access mode $X$,
such as \verb.SR. for synchronized read)
just before the programmer operates on the whole vector again.
It may be inefficient for a \texttt{pvector} to perform the \texttt{S$X$} synchronization
automatically, because multiple operations can be performed on the same \texttt{pvector} 
before accessing the whole \texttt{vector} again, 
and because the \texttt{pvector}
has no knowledge about when the operations on itself will be finished, hence
keeping them coherent each time is not necessary and thus a waste of performance.

Figure~\ref{fig:pvector}
shows an example of using a \texttt{pvector} and the 
\texttt{S$X$} macro.
It initializes a mother vector \texttt{vx}
and a \texttt{pvector} \texttt{vy} on it.
The following two lines change part of \texttt{vy}'s value multiple times. 
The \texttt{SR} macro explicitly restores coherence for
\texttt{vy} before the following read access, 
resulting in a write-back of \texttt{vy} elements in GPU device memory
to their locations in \texttt{vx},
thus also \texttt{vx} as a whole becomes coherent again
and line 6 is safe to operate on the whole \texttt{vx}.


\begin{figure}[tb]
\begin{minipage}{.48\textwidth}
\begin{small}
\begin{verbatim}
void coherent_on_cpu_r(){
  	 if( !cpu_coherent_unit ){
       download();
       cpu_coherent_unit=true;
  	 }
}
void coherent_on_cpu_w(){
  	 cpu_coherent_unit=true;
  	 gpu_coherent_unit=false;
}
void coherent_on_cpu_rw(){
  	 if( !cpu_coherent_unit ){
       download();
       cpu_coherent_unit=true;
  	 }
  	 gpu_coherent_unit=false;
}
\end{verbatim}
\end{small}
\end{minipage}\quad
\begin{minipage}{.48\textwidth}
\begin{small}
\begin{verbatim}
void coherent_on_gpu_r(){
  	 if( !gpu_coherent_unit ){
       upload();
       gpu_coherent_unit=true;
	 }
}
void coherent_on_gpu_w(){
  	 gpu_coherent_unit=true;
  	 cpu_coherent_unit=false;
}
void coherent_on_gpu_rw(){
  	 if( !gpu_coherent_unit ){
       upload();
       gpu_coherent_unit=true;
  	 }
  	 cpu_coherent_unit=false;
}
\end{verbatim}
\end{small}
\end{minipage}
\caption{\label{fig:vectorpucoherence}Coherence control code,
    here for simple vectors,
    in \texttt{vectorpu.h}. Functions \texttt{download} and
    \texttt{upload} are implemented using CUDA
    \texttt{thrust::copy}. Validity of copies on CPU and GPU
     is indicated
     by the flags \texttt{cpu\_coherent\_unit} and
     \texttt{gpu\_coherent\_unit} respectively; both are
     initialized to \texttt{true}
     in code allocating new vectors (not shown).}
     %
     %
\end{figure}

\subsection{Implementation Notes}

\paragraph{Coherence protocol}
In the source code of VectorPU, the code
relevant for our work is the part of \verb+vectorpu.h+%
\footnote{The VectorPU source code can be found at
\texttt{https://github.com/lilu09/vectorpu}.} 
that handles coherence. Its implementation for the various
variants of \texttt{vector} 
(see the code excerpt in Fig.~\ref{fig:vectorpucoherence} for
simple \texttt{vector}s) follows a
simple valid-invalid protocol.

\paragraph{Expansion of macro annotations to device-specific pointers}

For function parameters, the macro annotations expand into appropriate C code to fit their function call context.
For illustration, we show the simplified code after a function parameter's expansion ($\longrightarrow$) for four typical annotations,
where \texttt{vx} refers to a VectorPU \texttt{vector} instance:

\begin{small}
\begin{itemize}
 \item \texttt{R(vx)} $\longrightarrow$ 
 \begin{minipage}[t]{0.7\textwidth}
  \texttt{set\_coherence\_state();}\\
  \texttt{return this->std::vector<T>::data();}\\
  \texttt{//casted as const in return value}
 \end{minipage}
 \item \texttt{W(vx)} $\longrightarrow$ 
 \begin{minipage}[t]{0.7\textwidth} \texttt{set\_coherence\_state();}\\
 \texttt{return this->std::vector<T>::data();}
 \end{minipage}
 \item \texttt{GR(vx)} $\longrightarrow$ \begin{minipage}[t]{0.7\textwidth}
 \texttt{set\_coherence\_state();}\\
 \texttt{return thrust::raw\_pointer\_cast(}\\
 \hspace*{1em}\texttt{\& (* thrust::device\_vector<T>::begin() ) );}\\
 \texttt{//casted as const in return value}
 \end{minipage}
 \item \texttt{GW(vx)} $\longrightarrow$ 
\begin{minipage}[t]{0.7\textwidth} \texttt{set\_coherence\_state();\\ return thrust::raw\_pointer\_cast(}\\
\hspace*{1em}\texttt{\& (* thrust::device\_vector<T>::begin() ) );}
\end{minipage}
\end{itemize}
\end{small}

Hence, each annotated parameter is expanded to some code snippet\footnote{One can think of those code snippets as
anonymous functions or lambda functions. 
In the real scenarios these code snippets are encapsulated within a function,
and each macro as shown above expands to a call to its function.}.
In all scenarios the expanded macros first update the coherence state according to the annotation's semantics.
Then, for the CPU cases, a pointer to a \texttt{std::vector} is returned,
and for the GPU cases, a pointer to a Thrust pointer (which is a pointer to GPU memory space) is returned.
For read-only cases, the return value is casted to \texttt{const} to ensure type safety in its function invocation.
For write-only cases, no such \texttt{const} cast happens.

\paragraph{Partial vector implementation and memory management}
For implementing the \texttt{pvectors} atop \texttt{vector}s, 
VectorPU uses the simplistic approach of allocating memory for
the \emph{entire} \texttt{vector} on the device 
even if the \texttt{pvector}(s) might only
access a minor part of it. This may waste device memory space 
but makes local address calculations easy, and anyway only
the accessed elements (the \texttt{pvector} range) will be transferred.
As we will see later, it also simplifies coherence management
for \emph{overlapping} \texttt{pvector} accesses, which was not 
really foreseen 	in the original VectorPU design.


\subsection{Efficiency}

Using only available C++(11) language features, 
VectorPU provides a flexible unified memory
view where all data transfer and device memory management
is abstracted away from the programmer. Nevertheless,
its efficiency is on par with that of handwritten CUDA code
containing explicit data movement and memory management code
\cite{VectorPU-2017}.
In particular, the VectorPU prototype was shown to
achieve 1.4x to 13.29x speedup over good quality
code using Nvidia's \textit{Unified Memory} API
on several machines ranging from laptops to supercomputer nodes,
with Kepler and Maxwell GPUs. For a further discussion of
VectorPU features, such as
specialized versions of \verb-vector-, for descriptions
of how to use VectorPU together with lambda expressions 
e.g.\ to express skeleton computations, and for further
experimental results 
we refer to \cite{VectorPU-2017}.